\documentclass[11pt]{amsart}   

\usepackage{amssymb,amsthm,amsmath}
\newcommand{\oblock}[1]{\begin{matrix} #1 \end{matrix}}

\usepackage{color}
\definecolor{darkblue}{rgb}{0, 0, .4}
\definecolor{grey}{rgb}{.7, .7, .7}

\usepackage[breaklinks]{hyperref}
\hypersetup{
        colorlinks=true,
        linkcolor=darkblue,
        anchorcolor=darkblue,
        citecolor=darkblue,
        pagecolor=darkblue,
        urlcolor=darkblue,
        pdftitle={},
        pdfauthor={}
}

\newtheorem{theorem}{Theorem}[section]
\newtheorem{lemma}[theorem]{Lemma}

\theoremstyle{definition}
\newtheorem{definition}[theorem]{Definition}
\newtheorem{example}[theorem]{Example}

\theoremstyle{remark}
\newtheorem{remark}[theorem]{Remark}

\numberwithin{equation}{section}

\theoremstyle{theorem}
\newtheorem{corollary}[theorem]{Corollary}

\newtheorem{proposition}[theorem]{Proposition}

\begin{document}

\title[Using carry-truncated addition for add-rotate-xor algorithms]{Using carry-truncated addition to analyze add-rotate-xor hash algorithms}

\begin{abstract}
We introduce a truncated addition operation on pairs of $N$-bit binary numbers
that interpolates between ordinary addition mod $2^N$ and bitwise addition in
$\left(\mathbb{Z}/2\mathbb{Z}\right)^N$.  We use truncated addition to analyze
hash functions that are built from the bit operations add, rotate, and xor,
such as {\tt Blake}, {\tt Skein}, and {\tt Cubehash}.   Any ARX algorithm can
be approximated by replacing ordinary addition with truncated addition, and we
define a metric on such algorithms which we call the {\bf sensitivity}.  This
metric measures the smallest approximation agreeing with the full algorithm a
statistically useful portion of the time (we use $0.1\%$).  Because truncated
addition greatly reduces the complexity of the non-linear operation in ARX
algorithms, the approximated algorithms are more susceptible to both collision
and pre-image attacks, and we outline a potential collision attack explicitly.
We particularize some of these observations to the {\tt Skein} hash function. 
\end{abstract}

\author{Rebecca E. Field and Brant C. Jones}
\address{Department of Mathematics and Statistics, MSC 1911, James Madison University, Harrisonburg, VA 22807}
\email{\href{mailto:[fieldre,brant]@math.jmu.edu}{\texttt{[fieldre,brant]@math.jmu.edu}}}
\urladdr{\url{http://www.math.jmu.edu/\~[fieldre,brant]/}}


\maketitle

\bigskip
\section{Introduction}\label{s:background}

This paper is concerned with a family of hash algorithms that are defined in
terms of addition mod $2^{N}$ (denoted $+$), bitwise rotation, and exclusive or
(denoted $\oplus$) which is equivalent to bitwise addition mod $2$.  Such algorithms are 
referred to as ARX algorithms.

The non-linearity of ARX algorithms over $\left(\mathbb{Z}/2\mathbb{Z}\right)^N$
relies exclusively on the addition mod $2^N$ component.  As in base $10$, we can
perform addition on each digit and keep a carry value for each position to
record overflows.  In base $2$, we will observe that carrying occurs frequently
and so the addition-with-carrying operation is indeed highly non-linear and we
note that computers are designed to compute this type of non-linearity
efficiently.

In this work, we replace ordinary addition mod $2^N$ with a series of
approximations that converge to actual addition.  These approximations arise
from truncating the number of carry values that we record.  The zeroth
approximation is addition with no carries which corresponds to the exclusive-or
operation.  The first approximation is bitwise addition plus a single carry term
for each bit; namely, we look back a single bit for carry terms and do not
``carry our carries.''  The second approximation involves looking back two bits
for carry information, and so on.  The surprising fact is that for $64$-bit
binary numbers, the fourth approximation and the actual sum coincide a
statistically useful percentage of the time.  The eighth approximation coincides
with ordinary addition more than $90$ percent of the time.

In light of this, it is natural to consider replacing instances of ordinary
addition in an ARX algorithm by the simpler truncated addition operation.  We
describe a polynomial encoding for hash algorithms that can in principle be used
to find collisions and preimages for the algorithm with truncated addition.
Although neither attack is currently practical, we show that replacing ordinary
addition by truncated addition dramatically reduces the degree of these
polynomials, which should facilitate their analysis.  When collisions exist in
the version using truncated addition and the algorithm using truncated
addition agrees with the usual algorithm sufficiently often, then one obtains
collisions in the full algorithm with a significant nonzero probability.  

We also use this setting to describe a new metric that measures the strength of
ARX hash algorithms.  This metric can be described as the number of carry bits
that must be used before we can find cases where the full algorithm and its
approximation using truncated addition agree a statistically useful percent of
the time.  We measure this using a computer implementation of the algorithm and
a random search through 10 million inputs.  This metric is found to agree with
the popular wisdom, based on factors such as the speed of hashing, that {\tt
Cubehash}160+16/32+160-256 is stronger than the ARX algorithms that were final
round candidates for the SHA-3 competition.  This algorithm requires 13
bits of carrying before matches can be found.  In contrast, we were able to
find 29 cases of agreement per 10,000 random inputs using only 9 bits of
carrying for the algorithm {\tt Skein}.  This means that it suffices to attack
the $9$-truncated approximation rather than the full addition version of {\tt
Skein} as its approximation coincides a sufficient percent of the time.

The main technique in this paper, replacing addition with truncated addition, has been used as part of cryptographic attacks in the past.  In \cite{salsa}, a series of approximations for the hash algorithm {\tt Salsa20/8} (a reduced round version of the full {\tt Salsa} algorithm) are shown to possess the same bias in differential probabilities as the full 
algorithm.  As the full key is not necessary to trace backwards for the approximate algorithm, this differential bias can be used to distinguish key conjectures that are good candidates for the approximate to the true key.  Using a combination of second and third order approximation (two or three cary bits are recorded, but no others), the authors are able to show that a key can be found in a better than exhaustive search.

Here, we use truncated addition to define a new (and concrete) method to compare the robustness of different ARX algorithms.  As part of this comparison, each algorithm is assigned an approximation of sufficient complexity that any cryptographic attack can be applied to the approximations with statistically significant results for the full algorithm.  We also provide a direct combinatorial proof of the exact probability that truncated addition and ordinary addition will produce the same result, a significant improvement to approximations such as ``the $d^{th}$ order term may be ignored with probability $1-2^{-d}$" currently in the literature \cite{salsa}.

In Section~\ref{s:approx} we describe our truncated addition operation in
detail.  In Section~\ref{s:encoding} we explain how to encode a hash algorithm
as a system of polynomial equations.  Section~\ref{s:metric} gives some
empirical data about hash algorithms from the NIST competition \cite{nist}.  In
Section~\ref{s:skein} we give some suggestions for future
research regarding the algorithm {\tt Skein}.  A short conclusion follows in
Section~\ref{s:conclusions}.

\bigskip
\section{An approximation to addition by truncated carries}\label{s:approx}

Fix an integer $N$.  In our applications $N = 32$ or $N = 64$, and we represent
integers in binary notation using $N$ bits.  For example, $x = \sum_{i=0}^{N-1}
x_i 2^i$ has binary digits $x_i \in \{0,1\}$.  We will sometimes write these
digits as an array $[x_{N-1}, x_{N-2}, \ldots, x_1, x_0]$ with the least
significant bit in the rightmost position.

\begin{definition}
Let $x = \sum_{i=0}^{N-1} x_i 2^i$ and $y = \sum_{i=0}^{N-1} y_i 2^i$.  We can
then view $x$ and $y$ as elements of $\mathbb{Z}/(2^N \mathbb{Z})$ and
$\left(\mathbb{Z}/2\mathbb{Z}\right)^N$ simultaneously.  Here, $\mathbb{Z}/(2^N
\mathbb{Z})$ represents the group of integers with addition mod $2^N$, while
$\left(\mathbb{Z}/2\mathbb{Z}\right)^N$ represents bitstrings of length $N$
under componentwise addition mod $2$.  We denote the ordinary
addition of these integers in $\mathbb{Z}/(2^N \mathbb{Z})$ by $x+y$.  We denote
the bitwise addition of these integers 
\[ \sum_{i=0}^{N-1} \left( (x_i + y_i) \mod 2\right) 2^i \]
in $\left(\mathbb{Z}/2\mathbb{Z}\right)^N$ by $x \oplus y$.
\end{definition}

To relate these operations, we introduce the {\bf carry array} $\mathsf{c}(x,y)
= \sum_{i=1}^{N-1} \mathsf{c}_i(x,y) 2^i$, where
\[ \mathsf{c}_i(x,y) = \begin{cases} 1 & \text{ if $x_{i-1} + y_{i-1} +
    \mathsf{c}_{i-1}(x,y) \in \{2, 3\}$ } \\ 0 & \text{ otherwise. } \end{cases} \]
Then the usual addition algorithm using carries yields
\[ x + y = x \oplus y \oplus \mathsf{c}(x,y). \]

Observe that $\mathsf{c}_0(x,y)$ is always $0$ by definition.  If $x_{N-1} = 1$
and $y_{N-1} = 1$, then we would generate a carry at the $N$th position, but $2^N =
0$ in $\mathbb{Z}/(2^N \mathbb{Z})$ so we omit this.

\begin{lemma}\label{l:carry}
We have that $\mathsf{c}_i(x,y) = 1$ if and only if there exists $j < i$ such
that $(x_j, y_j) = (1,1)$ and for all $j < k < i$, we have $x_k + y_k = 1$.
\end{lemma}
\begin{proof}
It follows from the definitions that 
\[ \mathsf{c}_i(x,y) = \begin{cases}
    \mathsf{c}_{i-1}(x,y) & \text{ if $\mathsf{c}_{i-1}(x,y) = 1$ and $x_{i-1}+y_{i-1} \in \{1, 2\}$ } \\
    0 & \text{ if $\mathsf{c}_{i-1}(x,y) = 1$ and $x_{i-1}+y_{i-1} = 0$ } \\
    \mathsf{c}_{i-1}(x,y) & \text{ if $\mathsf{c}_{i-1}(x,y) = 0$ and $x_{i-1}+y_{i-1} \in \{0, 1\}$ } \\
    1 & \text{ if $\mathsf{c}_{i-1}(x,y) = 0$ and $x_{i-1}+y_{i-1} = 2$.} \\
\end{cases} \]
Hence, strings of carrying are started by a $(x_j,y_j) = (1,1)$ pair,
continued by $(0,1), (1,0)$ and $(1,1)$ pairs, and stopped by a $(0,0)$ pair.
If there are multiple $(1,1)$ pairs prior to position $i$, we choose the pair
with the greatest position $j$ so that $(x_k, y_k) \in \{ (0,1), (1,0) \}$ for
all $j < k < i$ by construction.
\end{proof}

Observe that in the worst case, we might have to look back $N-1$ positions to
decide whether a carry exists at the most significant position.  We now define a
version of addition based on a carry array that uses the information from at
most $m$ prior positions.

\begin{definition}\label{d:tradd}
Let $\mathsf{c}_i^{(m)}(x,y)$ be $1$ if there exists $i-m \leq j < i$ such that
$(x_j, y_j) = (1,1)$ and for all $j < k < i$ we have $x_k + y_k = 1$.
We then define the {\bf $m$-truncated addition} of $x$ and $y$ to be
\[ x +_{m} y := x \oplus y \oplus \mathsf{c}^{(m)}. \]
where $\mathsf{c}^{(m)} = \sum_{i = 1}^{N-1} \mathsf{c}_i^{(m)} 2^i$.
\end{definition}

Observe that $x +_0 y = x \oplus y$ and $x +_{(N-1)} y = x + y$ so truncated
addition generalizes and interpolates between these operations.

\begin{example}
If $N = 4$ then 
\[ \oblock{  & 1 & 0 & 0 & 1 \\ +_3 & 1 & 0 & 1 & 1 \\ \hline \\ & 0 & 1 & 0 & 0 \\} 
\hspace{0.9in} 
\oblock{  & 1 & 0 & 0 & 1 \\ +_1 & 1 & 0 & 1 & 1 \\ \hline \\ & 0 & 0 & 0 & 0 \\} \]
represents $9 + 11 = 20$ which is equivalent to $4 \mod 2^N$, and $9 +_1 11 =
0$, respectively.  In the first case where $m = N-1 = 3$, the carry array is $\mathsf{c}^{(3)}
= [0, 1, 1, 0]$.  In the second case where $m = 1$, the $1$-truncated carry
array is $\mathsf{c}^{(1)} = [0, 0, 1, 0]$.  We see that $\mathsf{c}^{(1)}_2 =
0$ since there is no $(1,1)$ pair lying within $m = 1$ positions prior to
position $i = 2$.  On the other hand, $\mathsf{c}^{(3)}_2 = 1$ since there does
exist a $(1,1)$ pair lying within $m = 3$ positions prior to position $i = 2$.
\end{example}

\begin{proposition}\label{p:av}
We have $x + y = x +_m y$ if and only if the sequence $\{x_i+y_i\}_{i=0}^{N-2}$
does not contain a $2$ directly followed by a contiguous subsequence of $m$
$1$'s as $i$ runs from $0$ to $N-2$.
\end{proposition}
\begin{proof}
This follows by comparing Lemma~\ref{l:carry} and Definition~\ref{d:tradd}.
\end{proof}

We are now in a position to determine the probability that $x +_m y$ agrees with
$x+y$.  Recall that a {\bf ternary string} is one in which each digit is $0$,
$1$ or $2$.

\begin{lemma}\label{l:table}
Let $P(m)$ be the ternary string $1^m 2 = 11\cdots12$.  Let $p_m(i)$ be the
probability that in a bitwise sum of uniformly chosen binary strings (of any
length $\geq m+1$), the rightmost instance of $P(m)$ as a consecutive substring ends at
position $i$.  Here, we label the positions from right to left, starting from
$0$.  Let $a_m(j)$ be the probability that a bitwise sum of uniformly chosen
binary strings of length $j$ does not contain $P(m)$ as a consecutive substring.
Then we have the system
\begin{equation}\label{e:am}
a_m(j) = 1 - \sum_{i = 0}^{(j-1)-m} p_m(i)
\end{equation}
\begin{equation}\label{e:pm}
p_m(i) = \left(\frac{1}{2}\right)^m \left(\frac{1}{4}\right) a_m(i)
\end{equation}
that can be solved explicitly for $a_m(N-1)$.
\end{lemma}
\begin{proof}
Every instance of $P(m)$ in a ternary string of length $j$ must end at some
position, and each such event is independent, so Equation~(\ref{e:am})
represents the probability that no instances of $P(m)$ occur.
Equation~(\ref{e:pm}) gives the probability that in the bitwise sum of two
uniformly chosen binary strings, the rightmost $i$ positions avoid $P(m)$,
the next position is a $2$ (this occurs with probability $1/4$), the $m$
subsequent positions are $1$'s (these each occur with probability $1/2$), and the
remaining positions are all unrestricted (so contribute probability 1).
\end{proof}

\begin{table}
\begin{center}
\begin{tabular}{|c|c|c|}
\hline
$m$ & $N = 32$-bit & $N = 64$-bit \\
\hline
4 & 63.62771 \% & 37.10136 \% \\
5 & 80.94266 \% & 62.31794 \% \\
6 & 90.49360 \% & 79.59719 \% \\
7 & 95.36429 \% & 89.50263 \% \\
8 & 97.76392 \% & 94.73115 \% \\
9 & 98.92764 \% & 97.38680 \% \\
10 & 99.48763 \% & 98.71143 \% \\
11 & 99.75591 \% & 99.36646 \% \\
12 & 99.88404 \% & 99.68900 \% \\
13 & 99.94507 \% & 99.84747 \% \\
14 & 99.97406 \% & 99.92525 \% \\
15 & 99.98779 \% & 99.96338 \% \\
16 & 99.99428 \% & 99.98207 \% \\
\hline
\end{tabular}
\caption{Probability of $x +_m y = x + y$}\label{t:pmatch}
\end{center}
\end{table}

\begin{corollary}\label{c:pragree}
The probability $\pi_m(N)$ that $x +_m y = x + y$ where $x$ and $y$ are
uniformly chosen $N$-bit integers is $a_m(N-1)$.  Some typical values of
$\pi_m(N)$ are illustrated in Table~\ref{t:pmatch}.
\end{corollary}
\begin{proof}
This follows from Proposition~\ref{p:av} and Lemma~\ref{l:table}.
\end{proof}

\bigskip
\section{A polynomial encoding and metrics for ARX algorithms}\label{s:encoding}

In this section, we consider encoding an ARX hash algorithm by a system of
polynomial functions over $\mathbb{F}_2$, the 2-element field.  Here, we mean
that the domain, range, and ring of coefficients of these polynomials should all
be $\mathbb{F}_2$.  We will see that replacing instances of $+$ by $+_m$ reduces
the degree of these polynomials, which facilitates analysis of the hash
algorithm.  At the same time, Table~\ref{t:pmatch} gives some evidence that
making this replacement will not change the output of the hash function too
often.

Observe that our $N$-bit arrays have an action of the symmetric group
$\mathcal{S}_N$ of permutations on $N$ letters given by permuting the entries of
arrays.  In particular, this action allows us to achieve the bitwise rotation
operation.  We denote this action by $\sigma \cdot [x_{N-1}, \ldots, x_0]$ for
$\sigma \in \mathcal{S}_N$.

\begin{proposition}\label{p:pencode}
Consider two $N$-bit arrays $x = [x_{N-1}, \ldots, x_1, x_0]$ and \\ $y = [y_{N-1},
\ldots, y_1, y_0]$, and let $\sigma \in \mathcal{S}_N$.  There exist polynomial
functions in \\
$\mathbb{F}_2[x_0, x_1, \ldots, x_{N-1}, y_0, y_1, \ldots,
y_{N-1}]$ whose evaluation is equal to the $i$th bit of $x \oplus y$, $x + y$
and $\sigma \cdot x$, respectively.  Explicitly, we have
\begin{itemize}
    \item  The $i$th bit of $\sigma \cdot [x_{N-1}, \ldots, x_1, x_0]$ is $x_{\sigma(i)}$.
    \item  The $i$th bit of $[x_{N-1}, \ldots, x_1, x_0] \oplus [y_{N-1},
        \ldots, y_1, y_0]$ is $x_i + y_i$.
    \item  The $i$th bit of $[x_{N-1}, \ldots, x_1, x_0] +_m [y_{N-1}, \ldots, y_1, y_0]$ is
        \[ (x_i + y_i) + \sum_{k = 1}^{\min(i,m)} (x_{i-k} y_{i-k}) \prod_{j=i-k+1}^{i-1}  (x_j + y_j). \]
\end{itemize}
\end{proposition}
\begin{proof}
The first two formulas are straightforward.  The last formula follows from
Definition~\ref{d:tradd}.
\end{proof}

\begin{example}
The addition of two $4$-bit numbers $[x_3, x_2, x_1, x_0] + [y_3, y_2, y_1,
y_0]$ can be represented by the polynomials
\[ [(x_3 + y_3) + (x_2 y_2) + (x_1 y_1)(x_2 + y_2) + (x_0 y_0)(x_1 + y_1)(x_2 +
y_2), \]
\[ (x_2 + y_2) + (x_1 y_1) + (x_0 y_0)(x_1 + y_1), 
 (x_1 + y_1) + (x_0 y_0), x_0 + y_0 ] \]
with maximum degree $4$.  If we use $2$-truncated addition instead, then we
obtain
\[ [(x_3 + y_3) + (x_2 y_2) + (x_1 y_1)(x_2 + y_2), (x_2 + y_2) + (x_1 y_1) + (x_0
y_0)(x_1 + y_1), \] \[ (x_1 + y_1) + (x_0 y_0), x_0 + y_0 ], \]
which has maximum degree $3$.
\end{example}

We consider an {\bf APX hash function} to be any finite composition of the
operations $+$, $\oplus$, and any permutation of the bits in an array.  To find
a collision for such a hash algorithm, it is helpful to have a message that is at least
as long as the output.  We therefore let $n$ be the maximum number of bits in
the input (including both the message as well as any key derived from the
message), output, or internal state.

Let $\bar{x}_i$ be variables representing the bits of input to the hash, so each
$\bar{x}_i \in \{0,1\}$ for $0 \leq i \leq n-1$.  We include variable bits for
the key if it is derived from the message.  We then use
Proposition~\ref{p:pencode} to build polynomials $\bar{y}_i \in
\mathbb{F}_2[\bar{x}_0, \bar{x}_1, \ldots, \bar{x}_{n-1}]$ that represent the
$i$th bit of output from the APX hash function.  We can encode multiple rounds
of a sub-algorithm by iterating the functions we obtain, taking the $\bar{y}_i$
expressed in terms of the $\bar{x}_i$ and using them as input.  

If we do this for two sets of inputs $\bar{x}_i$ and $\bar{x}_i'$, say, then
collisions correspond to nontrivial solutions of the system of polynomial
equations
\[ \{ \bar{y}_i(\bar{x}_0, \bar{x}_1, \ldots, \bar{x}_{n-1}) = \bar{y}_i(\bar{x}_0', \bar{x}_1', \ldots, \bar{x}_{n-1}') \}_{i=0}^{n-1}. \]
Similarly, if we let $\bar{z}_i$ be variables corresponding to the output of a
hash, then a preimage for the output $(\bar{z}_0,...,\bar{z}_{n-1})$
corresponds to a solution of the system of polynomial equations  
\[ \{ \bar{y}_i(\bar{x}_0, \bar{x}_1, \ldots, \bar{x}_{n-1}) = \bar{z}_i \}_{i=0}^{n-1}. \]

These systems each have $2n$ variables and all coefficients are $0$ or $1$.
Therefore, the maximal degree among the $y_i$ is a primary measure of the
complexity of this system, and hence of the APX algorithm.  Each $+$ operation
performed by the algorithm increases the degree, while bitwise permutations do
not increase it at all.

More precisely, we may observe that if $f$ and $g$ are polynomial functions that
represent single bits of output and $\deg(f) \geq \deg(g)$ then 
\[ \deg(f +_m g) = m \deg(f) + \deg(g) \]
by the equation given in Proposition~\ref{p:pencode}. 
Therefore, replacing $+ = +_{N-1}$ by $+_m$ dramatically reduces the degrees of
the encoding polynomials.

In principle, algorithms using Gr\"obner bases can be used to solve such systems
of polynomial equations, see e.g. \cite{ha}.  Neither the collision nor the preimage attacks we have
outlined seem to be currently practical, although this could change due
to an increase in computer power or more efficient Gr\"obner basis algorithms,
an active area of research in mathematics.

Although length of time to find a Gr\"obner basis is difficult to predict,
generally it is true that the higher the degree of the equations, the longer the
algorithm will take, so the degree of a hash algorithm gives a good measure of
algorithm complexity.  

\begin{definition}
We define the {\bf degree} of an APX hash function to be the maximum degree of
its encoding polynomials.
\end{definition}

For ARX algorithms, we have seen that this metric will be dominated by the
number of times $+$ is used in the algorithm.

\begin{definition}
Denote an ARX hash algorithm by $H$, and its output after hashing the message
$M$ by $H(M)$.  Given an ARX hash algorithm $H$, let $H_{m}$ denote the
corresponding algorithm in which all instances of $+$ have been replaced by
$+_m$.  We define the {\bf sensitivity} of $H$ to be the minimum $m$ such that
$H_{m}(M) = H(M)$ for at least $0.1$ percent of the inputs $M$ of each fixed
length.
\end{definition}

The sensitivity measures how vulnerable a given algorithm would be to the types
of attacks we have outlined above.  Notice that the degree and the
sensitivity are related because we would expect that an algorithm using $k$
addition operations would have $H_m(M) = H(M)$ with probability
$\left(\pi_m\right)^k$ by Corollary~\ref{c:pragree}.  This assumes that these
operations occur independently and that the distribution of inputs to the
addition operations are uniform.

\bigskip
\section{Examples from the NIST competition}\label{s:metric}

In this section, we use Monte Carlo experiments to estimate the sensitivity of
some NIST competition algorithms \cite{nist}.  We implemented versions of {\tt
Blake} \cite{blake} and {\tt Skein} \cite{skein} that use truncated addition,
and ran them using random inputs to determine how often these modified
algorithms agree with the original algorithm.  {\tt Cubehash} \cite{cubehash}
did not pass the second round of the NIST competition but also provides an
interesting example for analysis.  The results are displayed in
Table~\ref{t:results}.

\begin{table}
\begin{center}
\begin{tabular}{|c|c|c|c|c|c|}
\hline
Algorithm & Internal state size & Addition bits &  Sensitivity & Number of $+$ operations \\
\hline
{\tt Skein} & 256 & 64 & 9 & 278 \\ 
\hline
{\tt Blake} & 256 & 32 & 10  & 1345 \\
\hline
{\tt Cubehash} & 1024 & 32 & 13  & 6145 \\
\hline
\end{tabular}
\caption{Experimental results}\label{t:results}
\end{center}
\end{table}

These results were generated using 10 trials with 1,000,000 random inputs each.
For these trials, the match between {\tt Skein} using $+_8$ and {\tt
Skein} using $+$ was $.001\%$ while the match between {\tt Skein} using
$+_9$ and {\tt Skein} using $+$ was $.294\%$.  The match between {\tt
Blake} using $+_{10}$ and {\tt Blake} using $+$ was $.106\%$.  The match
between {\tt Cubehash} using $+_{13}$ and {\tt Cubehash} using $+$ was
$3.4319\%$ whereas we found no matches at all between {\tt Cubehash} using
$+_{12}$ and {\tt Cubehash} using $+$.

These results show that we may replace $+$ by the significantly simpler
operation $+_m$ (where $m = 9$, $10$, or $13$) and still achieve the same 
output at least $0.1\%$ of the time.  Therefore collisions found in the
truncated addition versions of the algorithms would translate to
collisions in the full algorithms a statistically useful percent of the time.

\begin{remark}
Since {\tt Blake} uses $32$-bit addition, our truncated approximation reduces
the degree of each addition from $32$ to degree $10$.  On the other hand, {\tt
Skein} uses $64$-bit addition so our truncated approximation gives a much
more dramatic reduction from degree $64$ to degree $9$.  For this reason, we
would say that {\tt Skein} is the weaker algorithm.
\end{remark}

\begin{remark}
There are a total of $278$ $+$ operations in {\tt Skein}.  If all of the
addition operations occurred in independently and in parallel, we would expect
the probability of a match between ${\tt Skein}_9$ (using $+_9$) and ${\tt
Skein}$ (using $+$) to be $\left(\pi_9(64)\right)^{278} = (0.97387)^{278} =
0.000635732714225483$.  In our Monte Carlo experiment, we actually found
matches with probability $0.00294$.

While there are permutations included in each round that amount to the addition
operations being in parallel, many of {\tt Skein}'s additions appear in series.
\end{remark}

\begin{remark}
{\tt Blake} has $1345$ total additions and sensitivity $10$, so we would expect
${\tt Blake}_{10}$ to match {\tt Blake} with probability
$\left(\pi_{10}(32)\right)^{1345} = (0.99488)^{1345} = 0.0010036724$.  In our
experiments, we actually found matches with probability $0.00106$.  This makes
{\tt Blake} almost perfectly efficient via our metric.
\end{remark}

\begin{remark}
The corresponding results for {\tt Cubehash} seem surprising.  The program we
used to compute the sensitivity of {\tt Cubehash} used only $6145$ $+$
operations.  (The number of operations in {\tt Cubehash} depends on the length of
the message being hashed, so it is important to not use generic figures for
this.)

We would expect ${\tt Cubehash}_{13}$ to match ${\tt Cubehash}$ with probability
$\left(\pi_{13}(32)\right)^{6145} = (0.99945)^{6145} = 0.0340243180867048$.  In our
experiments, we actually found matches {\em less} often, with probability
$0.00106$. 

To understand this result, note that differences between $+$ and $+_m$ arise from the
addition of two numbers with long strings of $0/1$ pairs in consecutive
entries.  If a hash algorithm were unlikely to turn inputs into their opposite
entry and then add the result to the original, then it is plausible to have such
a result.  In fact, unlike the other hash algorithms, {\tt Cubehash} uses only
odd rotation constants which may make it less likely to generate such strings.
\end{remark}

It would be interesting to understand the relationship between
$\pi_m(N)^{\text{number of $+$ operations}}$ and the experimental match
percentages more precisely.

\bigskip
\section{Future work for {\tt Skein}}\label{s:skein}

The heart of {\tt Skein} is the tweakable block cipher {\tt Threefish}, and it
is this cipher that we suggest analyzing using truncated addition.  The basic
structure of the {\tt Threefish} cipher is four applications of a non-linear
bijection (defined using add, rotate and xor operations) followed by the
addition of a full-length subkey.  More specifically, {\tt Threefish} breaks
the internal state of $256$ bits into two pairs of $64$-bit words and applies
to each pair an ARX function called {\tt MIX}.  After this, the four words are
permuted (the same permutation, ${\tt PERMUTE} = (0)(13)(2)$, being used each
time).  The rotation constants internal to {\tt MIX} are changed on a schedule
for optimal dispersal, and a `round' in {\tt Threefish} is the application of
one set of {\tt MIX}s and one {\tt PERMUTE}.  Every four rounds, a `subkey' of
length $256$ is added to the current state.  The full specification of {\tt
Threefish} calls for $72$ rounds, so $18$ subkeys added in total.  

Following the scheme outlined in Section~\ref{s:encoding}, a single round of
{\tt Threefish} can be made to act on a set of variables
\[ (x_0,...,x_{63},y_0,...,y_{63},z_0,...,z_{63},w_0,...,w_{63})=(\overline{x},\overline{y},\overline{z},\overline{w})={\bf x} \]
producing $256$ Boolean polynomials in the variables $x_0,...,w_{63}$, one
polynomial for each coordinate.  We call the $i$th such polynomial $f_i$ and
denote the full operation on all of these variables $f = (f_0, f_1, \ldots,
f_{256})$.  We similarly define the polynomials $f_i(x_0,...,w_{63})_m$ to be
the coordinate functions for the truncated addition version of ${\tt
Threefish}_m$ in which  all ordinary additions are replaced by $m$-truncated
addition.

Observe that $f$ is a bijection.  This is due to the fact that when any
add, rotate or xor operation within {\tt MIX} is applied to
$\overline{x},\overline{y}$, one of the two original inputs is retained.  This
follows from the definition
\[\text{{\tt MIX}}(\overline{x},\overline{y})=(\overline{x}+\overline{y},\rho(\overline{x})\oplus(\overline{x}+\overline{y})) \]
where $\rho$ is bitwise rotation.

We first consider the collision attack outlined in Section~\ref{s:encoding}.
Since there are no collisions if the step is a bijection, we must consider
non-bijective rounds.  As the non-bijectivity occurs from adding the subkey,
the first interesting computation would be:

\begin{quote}
Let $K_0$ be the first sub-key and $K_1$ be the second.  Let $I$ be the ideal generated by
\[ f(f(f(f(f({\bf x}+K_0)_m)_m)_m)_m+K_1)_m-f(f(f(f(f({\bf x}'+K_0)_m)_m)_m)_m+K_1)_m. \]
\end{quote}

A Gr\"obner basis for this ideal would detect the interaction between two
non-bijective rounds, yielding real information about the ${\tt Skein}_m$
algorithm.  Although we were unable to reverse enough rounds of ${\tt Skein}_m$
to make a practical attack, we did reverse two rounds of the $m=2$
carry-approximated algorithm on $16$-bits by computing a Gr\"obner
basis\footnote{Using {\tt Sage}/{\tt PolyBoRi} on a 2.53 GHz Intel Core i5
MacBook Pro.  We also investigated the $m=3$ carry-approximated algorithm on
$24$-bits for up to $3$ rounds {\tt Skein}.  While the number of polynomials is
always $24$, and the degrees of these polynomials do not exceed $16$, the
maximum number of terms in each polynomial grows from $10$ to $2521$ to $236187$
for $1$, $2$ and $3$ rounds of {\tt Skein}, respectively.  We attempted to find
a Gr\"obner basis for the ideal generated by these polynomials using {\tt
Sage}/{\tt PolyBoRi}, {\tt Macaulay 2}, and the {\tt Macaulay 2} package {\tt
BooleanGB} \cite{ha}, but none of these returned results for $2$ or more rounds.
These computations would be more feasible if a parallel version of the Gr\"obner
basis algorithm became available.}.

Next, we consider the preimage attack.   A preimage attack has no restrictions
on the number of rounds needed to be useful, as a preimage for even one round
is often difficult.  Let $I$ be the ideal generated by 
\[ {\bf z}-f({\bf x}+K_0)_m, \]
corresponding to the system of equations from Section~\ref{s:encoding}.  In
order to solve for ${\bf x}$ in terms of ${\bf z}$ and produce a true inverse
for one round of the algorithm with truncated addition, we will need to use a lex Gr\"obner
basis algorithm (with the variables in ${\bf z}<{\bf x}$) to produce an
elimination ideal.  As the rounds of ${\tt Threefish}_m$ are not identical (the
rotation constants are different for each round), an inverse for two rounds
would require the same analysis for the ideal generated by 
\[ {\bf z}-f(f({\bf x}+K_0)_m)_m, \]
and, theoretically, this process could be carried out for all $72$ rounds of
${\tt Threefish}_m$ where the rounds containing subkeys would force the introduction
of additional variables.  Although we do not have a practical attack, we were
able to reverse three rounds of the $m = 2$ carry-approximated algorithm on
$12$-bits by computing a Gr\"obner basis\footnote{Using {\tt Sage}/{\tt
PolyBoRi} on a 2.53 GHz Intel Core i5 MacBook Pro.  We were also able to reverse
one round of the $m = 2$ carry-approximated algorithm on $16$-bits.}.

We believe these approaches will lead to useful computations for others with more
computing resources to explore.

\bigskip
\section{Conclusions}\label{s:conclusions}

We have seen how to encode APX hash functions as systems of polynomials over
$\mathbb{F}_2$.  The degree of the approximation obtained by using $m$-truncated
addition will be significantly smaller than the degree of the original APX
function.  The sensitivity measures how small we can let $m$ be and still obtain
a function that reasonably approximates original APX hash function.

One open question that arises from this work is how to construct differential
attacks using the metrics we have described.  It would also be interesting to
examine the encoding polynomials for some of the NIST competition algorithms in
detail, and compute Gr\"obner bases for them.

\bigskip
\section*{Acknowledgments}

We thank Elizabeth Arnold for sharing her expertise on Gr\"obner basis
algorithms and Nicky Mouha for helpful comments on an earlier draft of this
work.  In addition, we'd like to acknowledge the anonymous reviewers who
provided valuable feedback.

\end{document}